\documentclass[9pt,twocolumn,twoside]{IEEEtran}

\usepackage{amsmath}
\usepackage{amssymb}
\usepackage{amsthm}
\usepackage{graphicx}
\usepackage{epstopdf}
\usepackage{bm}
\usepackage{color}
\usepackage{cite}
\usepackage{dsfont}
\usepackage{algorithm}
\usepackage[noend]{algpseudocode}
\usepackage{footnote}

\newtheorem{theorem}{Theorem}
\newtheorem{lemma}{Lemma}

\newtheorem{definition}{Definition}

\newtheorem{remark}{Remark}

\newcommand{\myexpect}[1]{\mathbb{E}[#1]}

\newcommand{\tr}{\text{tr}}

\begin{document}

\title{Stability Enforced Bandit Algorithms for Channel Selection in Remote State Estimation of Gauss-Markov Processes}
\author{Alex S.\ Leong, \IEEEmembership{Member, IEEE}, Daniel E.\ Quevedo, \IEEEmembership{Fellow, IEEE}, and Wanchun Liu, \IEEEmembership{Member, IEEE} \\  
 \thanks{A.\ S.\ Leong is with Defence Science and Technology Group, Melbourne, Australia. E-mail: {\tt alex.leong@defence.gov.au.} D.~E.~Quevedo is with the
  School of Electrical Engineering and Robotics,
  Queensland University of Technology, Brisbane, Australia. E-mail: {\tt dquevedo@ieee.org}. W.\ Liu is with the School
of Electrical and Information Engineering, University of Sydney, Sydney, Australia. Email: {\tt wanchun.liu@sydney.edu.au.}}
  }

\maketitle

\begin{abstract}
In this paper we consider the  problem of remote state estimation  of a Gauss-Markov process,  where a sensor can, at each discrete time instant, transmit on one out of \textit{M} different communication channels. A key difficulty of the situation at hand is that  the channel statistics are unknown. We study the case where both learning of the channel reception probabilities and state estimation is carried out simultaneously. Methods for choosing the channels based on techniques for multi-armed bandits are presented, and shown to provide stability. Furthermore, we define the performance notion of estimation regret, and derive bounds on how it scales with time for the considered algorithms. 
\end{abstract}


\section{Introduction}
The use of machine learning techniques for estimation and control has gained increasing attention in recent years \cite{Busoniu,Hewing_review,special_issue_CSS_letters}. Learning based estimation and control holds the promise of enabling the solution of  problems which are difficult or even intractable using traditional control design techniques. 

In this paper we focus on the  problem of remote state estimation of a Gauss-Markov process,  where the channel statistics are  constant but unknown.  As a motivating example, consider the use of unmanned aerial vehicles (UAVs) for providing wireless communication capabilities \cite{ZengZhangLim,Mohamed_UAV}, with the UAV acting as a mobile aerial base station collecting information from wireless sensors. The use of such systems can allow for monitoring of processes which are difficult to access otherwise, e.g., in a hostile or contested environment. The UAV would transmit collected information to a remote estimator via one of $M$ different channels, e.g., using $M$ different frequency bands. Characteristics of the individual channels can vary substantially at different locations/environments~\cite{ZengZhangLim,Li_MAB_survey}, which makes it difficult to have accurate knowledge of the channel statistics whenever the UAV moves to a different location. 
In this paper we will assume that, at any given location, the channel statistics are constant but unknown.

Learning of transmission schedules for a single system over a channel with unknown packet reception probability has been studied in~\cite{WuRen_learning}, while learning of power allocations for multiple control systems connected over an unknown non-stationary channel is considered in~\cite{Eisen_TSP}. 
With multiple channels, channel allocation for multiple processes using deep reinforcement learning  have been considered in \cite{LeongRamaswamy,RedderRamaswamy,liu2021DRL}. While the use of deep reinforcement learning techniques can find schedules which perform well, in general stabilizing properties of the learned policies are not guaranteed. Furthermore, many training samples may be needed during learning. 
In this paper, we consider a remote state estimation where, at every time instant,  a sensor or controller can choose one of $M$ unknown channels for transmission. The aim is to  learn the best channel while guaranteeing stability of the system and being \emph{sample efficient} (via maintaining  \emph{low regret}). 

We will view the channel selection task as a multi-armed bandit type problem \cite{Slivkins_book,Lattimore_book}. The multi-armed bandit problem is a simple example of a reinforcement learning problem which  captures the exploration vs.\ exploitation tradeoff, between basing decisions only on knowledge currently obtained about a system (``exploit'') vs.\ trying out alternative decisions which may potentially be better (``explore''). 
Among many applications in a wide variety of domains, multi-armed bandit techniques have been used to address scheduling problems in wireless communications \cite{GaiKrishnamachariJain_DySPAN,Li_MAB_survey,StahlbuhkShraderModiano,TakeuchiHasegawa}.

In contrast to classical multi-armed bandit problems, here the underlying processes are dynamical systems, which raises additional challenges on stability and performance. In this paper, we will study the use of the $\varepsilon$-greedy algorithm \cite{SuttonBarto}, as well as several algorithms based on the Thompson Sampling technique \cite{Russo_tutorial}, to carry out channel selection in the context of remote state estimation. 
Thompson Sampling is a sampling based algorithm for the multi-armed bandit problem, which has been shown to be optimal with respect to the accumulated regret in certain problems. 
Interestingly, Thompson sampling has  been previously used to address problems in stochastic \cite{Kim_TS,BanjevicKim} and linear quadratic control~\cite{OuyangGagraniJain}. 

\emph{Summary of contributions}:
The main contributions of the current work are as follows: 
\begin{itemize}
\item We propose the use of bandit algorithms for channel selection in the context of remote state estimation, including a new stability-aware Bayesian sampling algorithm. 
\item We prove that the considered algorithms can guarantee stability of the remote estimation process. 
\item We study performance of these algorithms by introducing the notion of estimation  regret, and prove that the $\varepsilon$-greedy algorithm has linear estimation regret, while the sampling based algorithms can achieve logarithmic estimation regret. 
\end{itemize}

The remainder of the paper is organized as follows: The system model is provided in Section \ref{sec:system_model}. A selection of bandit algorithms for channel selection are presented in Section \ref{sec:bandit_algorithms}. Stability and performance bounds of these algorithms are studied in Sections \ref{sec:stability} and \ref{sec:regret_bounds} respectively. 
Numerical studies are given in Section \ref{sec:numerical}. Section~\ref{sec:conclusion} draws conclusions.

\emph{Notation}: Given a matrix $X$, we say that $X \geq 0$ if $X$ is positive semi-definite. For matrices $X$ and $Y$, we say that $X \leq Y$ if $Y - X \geq 0$. The spectral radius of a matrix $X$ is denoted by $\rho(X)$. The beta distribution with parameters $\alpha$ and $\beta$ is denoted by $\textnormal{Beta}(\alpha,\beta)$, with expected value $\alpha/(\alpha+\beta)$. $\mathds{1}(\cdot)$ is the indicator function that returns 1 if its argument is true and 0 otherwise. 

\section{System Model}
\label{sec:system_model}
 \begin{figure}[t!]
\centering 
\includegraphics[scale=0.4]{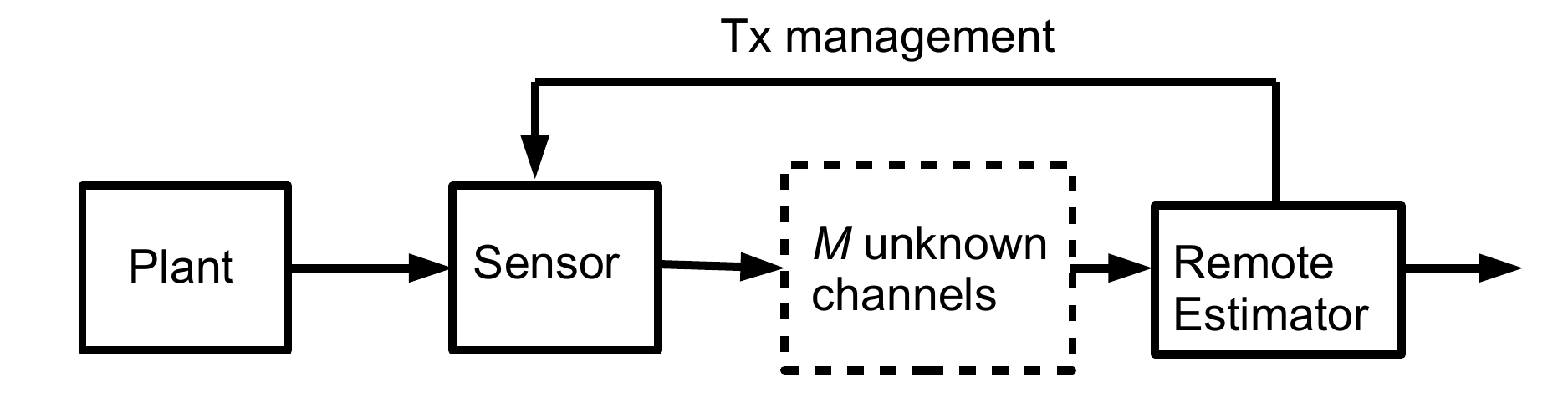} 
\caption{Remote state estimation over unknown channels}
\label{fig:system_model}
\end{figure} 
We consider a discrete-time plant  described by a Gauss-Markov process: 
\begin{equation}
    \label{eqn:sys}
    x_{k+1} = A x_k + w_k,
\end{equation}
with sensor measurements
$$y_k = C x_k + v_k,$$
where $w_k \sim \mathcal{N}(0,Q)$, $v_k \sim \mathcal{N}(0,R)$, see Fig.\ \ref{fig:system_model}.
The sensor   runs a local Kalman filter, with local state estimates $\hat{x}_k^s \triangleq \mathbb{E}[x_k|y_0,\dots,y_k]$ and estimation error covariances $P_k^s \triangleq \mathbb{E}[(x_k - \hat{x}_k^s)(x_k - \hat{x}_k^s)^T|y_0,\dots,y_k]$. We assume that the pair $(A,C)$ is observable and the pair $(A, Q^{1/2})$ is controllable, which implies that $P_k^s$ converges to some $\overline{P}$ as $k \rightarrow \infty$ \cite{andersonMoore}. We will assume steady state, so that $P_k^s = \overline{P}, \forall k$.

As depicted in Fig.\ \ref{fig:system_model}, the local state estimates are transmitted over lossy channels to a remote estimator. There are $M$ such channels, each i.i.d. Bernoulli with packet reception probabilities $\theta_m, m=1,\dots,M$. For notational simplicity, we will assume that all the packet reception probabilities are different, i.e., $\theta_i \neq \theta_j, \forall i\neq j$.  

\par At every time step $k$, we can choose one of the $M$ channels to transmit over. This decision is determined by the remote estimator, which notifies the sensor via a transmission (Tx) management channel. As the information transmitted on the transmission management channel is normally only a few bits in length, we assume that these transmissions are error-free, see also \cite{LeongRamaswamy}. 

A key issue is  that none of the packet reception probabilities are known. Ideally, we would like to always use the best channel (i.e., the channel with the highest packet reception probability). But since the packet reception probabilities are unknown, this requires us to find/learn the best channel by trying out different channels and observing their outcomes. The situation  resembles a multi-armed bandit type problem \cite{Slivkins_book,Lattimore_book}, with the channels acting as the different arms. The difference between the situation considered in this paper and classical multi-armed bandit problems is that  the underlying estimation process has dynamics. This raises additional issues such as stability, as investigated in the current work.

Before proceeding, we recall (see, e.g., \cite{ShiEpsteinMurray}) that the remote estimator has estimation error covariance
\begin{equation}
\label{remote_estimator_eqns}
P_k = \left\{\begin{array}{cc} \overline{P}, & \gamma_k = 1 \\ h(P_{k-1}), & \gamma_k=0  \end{array} \right.
\end{equation}
where 
\begin{equation}
\label{eqn:h_defn}
h(X) \triangleq A X A^\top + Q.
\end{equation}
In the above, $\gamma_k = 1$ means that the transmission at time $k$ was successful, while $\gamma_k=0$ means that the transmission was lost. We assume that $P_0 = \overline{P}.$

\par When $A$ is unstable, we know that if the same channel $m$ is chosen at every time step, then the remote estimator has bounded expected estimation error covariance for all $k$ if and only if its packet reception probability satisfies $\theta_m > \theta_c$ \cite{XuHespanha,Schenato}, where 
\begin{equation}
    \label{eqn:crit}
    \theta_c \triangleq 1 - \frac{1}{\rho(A)^2}.
\end{equation} 
Accordingly, in this paper we make the assumption that at least one of the channels has packet reception probability larger than $\theta_c$, i.e., 
\begin{equation}
\label{best_channel_condition}
\theta^* \triangleq \max_{m} \theta_m > \theta_c.
\end{equation}

\subsection{Preliminaries}

For later reference, we mention some properties of the operator $h(.)$ defined in \eqref{eqn:h_defn}.
\begin{lemma}
\label{lemma:h_properties}
The operator $h(.)$ defined in \eqref{eqn:h_defn} satisfies: \\
(i) $h(X) \leq h(Y) $ if $X \leq Y$ \\
(ii) $h(\overline{P}) \geq \overline{P}$\\
(iii) $\textnormal{tr} (h(\overline{P})) > \textnormal{tr} (\overline{P})$.
\end{lemma}
\begin{proof}
The proof of (i) follows easily from the definition~\eqref{eqn:h_defn}. Proofs of (ii) and (iii) can be found in \cite{ShiZhang}.
\end{proof}

Note that Lemma \ref{lemma:h_properties} implies that $P_k \geq \overline{P}, \forall k$, for $P_k$ evolving according to \eqref{remote_estimator_eqns}. 

The following recursion \eqref{eqn:EP_recursion} for the expected error covariances will also be used later. Let $\theta(k)$ be the reception probability at time $k$.  From \eqref{remote_estimator_eqns} we have
\begin{align*}
\mathbb{E}[P_{k+1}|P_k] & = \mathbb{E}[P_{k+1}|P_k, \gamma_k = 1] \mathbb{P}(\gamma_k=1) \\ & \quad + \mathbb{E}[P_{k+1}|P_k, \gamma_k = 0] \mathbb{P}(\gamma_k=0)\\
& =  \theta(k) \overline{P} + (1-\theta(k)) (A P_k A^\top + Q). 
\end{align*}
Then, since $h(.)$ is linear,  we can obtain
\begin{equation}
\label{eqn:EP_recursion}
\mathbb{E}[P_{k+1}] = \mathbb{E}[\mathbb{E}[P_{k+1}|P_k]] =  \theta(k) \overline{P} + (1-\theta(k)) h(\mathbb{E}[P_k]).
\end{equation}

\section{Bandit Algorithms for Channel Selection}
\label{sec:bandit_algorithms}
In this section we present a selection of multi-armed bandit algorithms for carrying out  channel selection, mainly based on the $\varepsilon$-greedy and Thompson sampling approaches \cite{Russo_tutorial}. Thompson sampling  is computationally efficient and can be easily implemented, requiring very little tuning apart from the choice of a prior distribution. There exist other bandit algorithms such as those based on the upper confidence bound (UCB) approach, which are not considered in the current work. Numerical studies however have shown that  in many situations the performance of Thompson sampling  is highly competitive against even well-tuned UCB algorithms~\cite{Granmo,MayKorda,Fatale_AOI}. 
\begin{savenotes}
\begin{algorithm}[t]
\caption{$\varepsilon$-greedy}
\label{alg:epsilon_greedy}
\begin{algorithmic}[1]
\State \textbf{Initialize}: $\alpha_{m,1} = 1, \beta_{m,1} = 1, \quad m=1,\dots,M$
\For{$k = 1, 2, \dots$}
    \State Set $\hat{\theta}_{m,k}^{\varepsilon} = \frac{\alpha_{m,k}}{\alpha_{m,k}+\beta_{m,k}}, \quad m=1,\dots,M$
    \State With probability $\varepsilon > 0$, choose $m_{k}^{\varepsilon}$ uniformly from $\{1,\dots,M\},$
    else choose\footnote{If multiple channels have the same maximum estimated packet reception probability, then a channel is chosen uniformly at random from these (maximizing) channels.} $m_{k}^{\varepsilon} = \textnormal{arg}\max_{m} \hat{\theta}_{m,k}^{\varepsilon}$ 
    \State Observe $\gamma_k$
    \State Update sample means via
    \begin{align*}
        \alpha_{m,k+1} & = \alpha_{m,k} + \gamma_k, & \textnormal{ if } m = m_{k}^{\varepsilon} \\
        \beta_{m,k+1} &= \beta_{m,k} + 1 - \gamma_k, & \textnormal{ if } m = m_{k}^{\varepsilon} \\
        \alpha_{m,k+1} & = \alpha_{m,k},  \, 
        \beta_{m,k+1} = \beta_{m,k}, & \textnormal{ if } m \neq m_{k}^{\varepsilon}
    \end{align*}
\EndFor
\end{algorithmic}
\end{algorithm} 
\end{savenotes}

\subsection{$\varepsilon$-greedy}
\label{sec:epsilon_greedy}
We first consider the well-known $\varepsilon$-greedy algorithm \cite{SuttonBarto}, which with probability $\varepsilon > 0$ chooses a channel uniformly at random (``explore''), and with probability $1-\varepsilon$ chooses the channel with the maximum estimated (via a sample mean) packet reception probability (``exploit''). After transmitting over this channel and observing $\gamma_k$, i.e., whether the transmission was successful, we update the sample means of the packet reception probabilities. The procedure is stated  in Algorithm \ref{alg:epsilon_greedy}, where we have  expressed the estimated packet reception probabilities in terms of the quantities $\alpha_{m,k}$ and $\beta_{m,k}$. These can be regarded as pseudo-counts of the number of successful and non-successful transmissions  respectively on channel $m$ \cite{Russo_tutorial}, and are updated using the formulas $\alpha_{m,k+1}  = \alpha_{m,k} + \gamma_k$ and $\beta_{m,k+1} = \beta_{m,k} + 1 - \gamma_k$.  This allows us to emphasize   similarities in the implementation with the other algorithms considered below. 

\subsection{Thompson Sampling}
Thompson sampling is a sampling based approach which draws samples $\hat{\theta}_{m,k}^{\textnormal{TS}}, m=1,\dots,M$ from the posterior distribution of the unknown packet reception probability of each arm/channel. The posterior distribution can be shown to have a beta distribution if the prior distribution is beta distributed~\cite{Russo_tutorial}, as we will assume here. In particular, the prior beta distribution with $\alpha_{m,1} = 1, \beta_{m,1} = 1, \, m=1,\dots,M$ is uniformly distributed in $[0,1]$. The channel that is selected is then given by $$m^{\textnormal{TS}}_k = \textnormal{arg}\max_{m} \hat{\theta}_{m,k}^{\textnormal{TS}}.$$
After transmitting over this channel and observing $\gamma_k$,  the posterior  is then updated to a $\textnormal{Beta}(\alpha_{m,k+1}, \beta_{m,k+1})$ distribution, using the same formulas for updating $\alpha_{m,k}$ and $\beta_{m,k}$ as in Section \ref{sec:epsilon_greedy}.  The procedure is summarized in Algorithm \ref{alg:TS}. 

\begin{algorithm}[t]
\caption{Thompson Sampling (TS)}
\label{alg:TS}
\begin{algorithmic}[1]
\State \textbf{Initialize}: $\alpha_{m,1} = 1, \beta_{m,1} = 1, \quad m=1,\dots,M$
\For{$k = 1, 2, \dots$}
    \State Sample $\hat{\theta}_{m,k}^{\textnormal{TS}} \sim \textnormal{Beta}(\alpha_{m,k}, \beta_{m,k}), \quad m=1,\dots,M$
    \State Choose $m_{k}^{\textnormal{TS}} = \textnormal{arg}\max_{m} \hat{\theta}_{m,k}^{\textnormal{TS}}$ and observe $\gamma_k$
    \State Update posteriors to $\textnormal{Beta}(\alpha_{m,k+1}, \beta_{m,k+1})$ via 
    \begin{align*}
        \alpha_{m,k+1} & = \alpha_{m,k} + \gamma_k, & \textnormal{ if } m = m_{k}^{\textnormal{TS}} \\
        \beta_{m,k+1} &= \beta_{m,k} + 1 - \gamma_k, & \textnormal{ if } m = m_{k}^{\textnormal{TS}} \\
        \alpha_{m,k+1} & = \alpha_{m,k},  \, 
        \beta_{m,k+1} = \beta_{m,k}, & \textnormal{ if } m \neq m_{k}^{\textnormal{TS}}
    \end{align*}
\EndFor
\end{algorithmic}
\end{algorithm} 

\subsection{Optimistic Bayesian Sampling}
In \cite{MayKorda} a modification of Thompson sampling called ``Optimistic Bayesian sampling'' (OBS) was proposed and studied. Recall that in Thompson sampling, at time $k$ one would draw samples $\hat{\theta}_{m,k}^{\textnormal{TS}}, m=1,\dots,M$ and then select $m^{\textnormal{TS}}_k = \textnormal{arg}\max_{m} \hat{\theta}_{m,k}^{\textnormal{TS}}.$
In OBS we would also draw  samples $\hat{\theta}_{m,k}^{\textnormal{TS}}, m=1,\dots,M$, but now select
\begin{align*}
m^{\textnormal{OBS}}_k &= \textnormal{arg}\max_{m} \hat{\theta}_{m,k}^{\textnormal{OBS}},
\end{align*}
where
\begin{align*}
\hat{\theta}_{m,k}^{\textnormal{OBS}} & \triangleq  \max\left(\hat{\theta}_{m,k}^{\textnormal{TS}}, \mathbb{E}[\hat{\theta}_{m,k}^{\textnormal{TS}}]   \right)  =  \max\left(\hat{\theta}_{m,k}^{\textnormal{TS}}, \frac{\alpha_{m,k}}{\alpha_{m,k} \!+\! \beta_{m,k}}   \right). 
\end{align*}
The procedure is summarized in Algorithm \ref{alg:OBS}. 
\begin{algorithm}[t]
\caption{Optimistic Bayesian Sampling (OBS)}
\label{alg:OBS}
\begin{algorithmic}[1]
\State \textbf{Initialize}: $\alpha_{m,1} = 1, \beta_{m,1} = 1, \quad m=1,\dots,M$
\For{$k = 1, 2, \dots$}
    \State Sample $\hat{\theta}_{m,k}^{\textnormal{TS}} \sim \textnormal{Beta}(\alpha_{m,k}, \beta_{m,k}), \quad m=1,\dots,M$
    \State Set $\hat{\theta}_{m,k}^{\textnormal{OBS}} = \max\left(\hat{\theta}_{m,k}^{\textnormal{TS}}, \frac{\alpha_{m,k}}{\alpha_{m,k} \!+\! \beta_{m,k}}   \right)$
    \State Choose $m_{k}^{\textnormal{OBS}} = \textnormal{arg}\max_{m} \hat{\theta}_{m,k}^{\textnormal{OBS}}$ and observe $\gamma_k$
    \State Update posteriors  to $\textnormal{Beta}(\alpha_{m,k+1}, \beta_{m,k+1})$ via 
    \begin{align*}
        \alpha_{m,k+1} & = \alpha_{m,k} + \gamma_k, & \textnormal{ if } m = m_{k}^{\textnormal{OBS}} \\
        \beta_{m,k+1} &= \beta_{m,k} + 1 - \gamma_k, & \textnormal{ if } m = m_{k}^{\textnormal{OBS}} \\
        \alpha_{m,k+1} & = \alpha_{m,k},  \, 
        \beta_{m,k+1} = \beta_{m,k}, & \textnormal{ if } m \neq m_{k}^{\textnormal{OBS}}
    \end{align*}
\EndFor
\end{algorithmic}
\end{algorithm} 
The motivation for OBS is to result in increased selection probabilities for arms which are uncertain \cite{MayKorda}. For classical multi-armed bandit problems, simulations have shown that OBS can outperform Thompson sampling in certain situations \cite{MayKorda,Mellor_thesis}.

\subsection{Stability-aware Bayesian Sampling}
In the context of selecting channels for the purpose of remote state estimation of dynamical systems as in \eqref{eqn:sys}, it is important to take into account the long term estimation performance. In particular, the fundamental bound \eqref{eqn:crit} suggests that, unlike when using OBS,  one should not artificially stimulate the selection probabilities for channels which seem to be poor. More specifically, if the current mean-estimate $\alpha_{m,k}/(\alpha_{m,k} + \beta_{m,k})$ of the reception probability is less than the critical threshold $\theta_c$, then it is questionable whether  this channel should be chosen with a probability higher than that prescribed by Thompson sampling. 

Motivated by the above, in this paper, we propose a \emph{stability-aware modification} of OBS which uses the OBS sample  $\hat{\theta}_{m,k}^{\textnormal{OBS}} $ only if $\alpha_{m,k}/(\alpha_{m,k} + \beta_{m,k}) > \theta_c$, while using the Thompson sample $\hat{\theta}_{m,k}^{\textnormal{TS}}$ otherwise. 
The algorithm, here named \emph{Stability-aware Bayesian Sampling} (SBS), can thus be characterized by: 
\begin{equation}
\label{eqn:OBS_mod}
\hat{\theta}_{m,k}^{\textnormal{SBS}} = \eta_{m,k} \hat{\theta}_{m,k}^{\textnormal{OBS}} + (1-\eta_{m,k}) \hat{\theta}_{m,k}^{\textnormal{TS}}, \quad m=1, \dots,M    
\end{equation}
and 
$$m_k^{\textnormal{SBS}} = \arg\max_m \hat{\theta}_{m,k}^{\textnormal{SBS}},$$
where \begin{equation}
    \label{eqn:obs_mod}
    \eta_{m,k} =
    \begin{cases}
    1&\text{if $\alpha_{m,k}/(\alpha_{m,k} + \beta_{m,k}) > \theta_c$}\\
    0&\text{otherwise.}
    \end{cases}
\end{equation}
The method is summarized in Algorithm \ref{alg:SBS}. 
\begin{algorithm}[t]

\caption{Stability-aware Bayesian Sampling (SBS)}
\label{alg:SBS}
\begin{algorithmic}[1]
\State \textbf{Initialize}: $\alpha_{m,1} = 1, \beta_{m,1} = 1, \quad m=1,\dots,M$
\For{$k = 1, 2, \dots$}
    \State Sample $\hat{\theta}_{m,k}^{\textnormal{TS}} \sim \textnormal{Beta}(\alpha_{m,k}, \beta_{m,k}), \quad m=1,\dots,M$
    \State Set $\hat{\theta}_{m,k}^{\textnormal{SBS}} = \eta_{m,k} \hat{\theta}_{m,k}^{\textnormal{OBS}} + (1-\eta_{m,k}) \hat{\theta}_{m,k}^{\textnormal{TS}}$, where $\hat{\theta}_{m,k}^{\textnormal{OBS}} = \max\left(\hat{\theta}_{m,k}^{\textnormal{TS}}, \frac{\alpha_{m,k}}{\alpha_{m,k} \!+\! \beta_{m,k}}   \right)$ and $\eta_{m,k} = \mathds{1}(\alpha_{m,k}/(\alpha_{m,k} + \beta_{m,k}) > \theta_c)$
    \State Choose $m_{k}^{\textnormal{SBS}} = \textnormal{arg}\max_{m} \hat{\theta}_{m,k}^{\textnormal{SBS}}$ and observe $\gamma_k$
    \State Update posteriors  to $\textnormal{Beta}(\alpha_{m,k+1}, \beta_{m,k+1})$ via 
    \begin{align*}
        \alpha_{m,k+1} & = \alpha_{m,k} + \gamma_k, & \textnormal{ if } m = m_{k}^{\textnormal{SBS}} \\
        \beta_{m,k+1} &= \beta_{m,k} + 1 - \gamma_k, & \textnormal{ if } m = m_{k}^{\textnormal{SBS}} \\
        \alpha_{m,k+1} & = \alpha_{m,k},  \,
        \beta_{m,k+1} = \beta_{m,k}, & \textnormal{ if } m \neq m_{k}^{\textnormal{SBS}}
    \end{align*}
\EndFor
\end{algorithmic}
\end{algorithm}

\section{Stability of Remote State Estimation with Channel Selection}
\label{sec:stability}
In this section, we show that the $\varepsilon$-greedy algorithm will ensure stability of the remote estimator, provided the exploration rate $\varepsilon$ is not too high, while all of the sampling based channel selection schemes presented in the previous section (Thompson sampling, OBS, and SBS) will ensure stability.  It is important to emphasize that Theorems  \ref{thm:stability_epsilon_greedy} and \ref{thm:stability_TS}, stated below, cover the   challenging case where channel qualities are unknown, and where channel selection and learning are done in real-time. 

Let $m_k$ be the channel chosen at time $k$, and define $$m^* \triangleq \textnormal{argmax}_m \theta_m.$$

\begin{theorem}
\label{thm:stability_epsilon_greedy}
Suppose condition \eqref{best_channel_condition} holds. Then under the $\varepsilon$-greedy algorithm, the expected estimation error covariance $\mathbb{E}[P_k]$ is bounded for all $k$ if and only if \begin{equation} 
\label{epsilon_stability_condition}
0 < \varepsilon < \frac{\theta^* - \theta_c}{\theta^* - \frac{1}{M}\sum_{m=1}^M \theta_m}.
\end{equation}
\end{theorem}

\begin{proof}
Because of the exploration phase, where with probability $\varepsilon > 0$ a channel is chosen randomly, it is easy to see that each channel will be accessed infinitely often. Hence the estimates of the packet reception probabilities for each of the channels will converge to their true values as $k \rightarrow \infty$. Thus, during the exploitation phase we have 
$$\mathbb{P}(m_k = m^* | \textnormal{exploit}) \rightarrow 1 \textnormal{ as } k \rightarrow \infty,$$ 
while during the exploration phase we have
$$\mathbb{P}(m_k = m^* | \textnormal{explore}) = \frac{1}{M}.$$
Then, given any $\delta > 0$, there exists a $K(\delta) < \infty$ such that
$$\mathbb{P}(\textnormal{reception at time } k) \geq (1-\varepsilon)\theta^*(1-\delta) + \frac{\varepsilon}{M} \sum_{m=1}^M \theta_m$$ 
for all $k > K(\delta)$.
Now let $$\theta' \triangleq (1-\varepsilon)\theta^*(1-\delta) + \frac{\varepsilon}{M} \sum_{m=1}^M \theta_m.$$ For $k > K(\delta)$, we have by Lemma \ref{lemma:h_properties} and a similar derivation to \eqref{eqn:EP_recursion} 
$$ \mathbb{E}[P_{k+1}] \leq \theta' \overline{P} + (1-\theta') h(\mathbb{E}[P_k]).$$
By induction, it follows that $\{\mathbb{E}[P_k]\}$ will be upper bounded by a sequence $\{V_k\}$ defined by
\begin{align*}
V_k & = \mathbb{E}[P_k], \quad k \leq K(\delta) \\
V_{k+1} & = \theta'  \overline{P} + (1-\theta') h(V_k), \quad k \geq K(\delta).
\end{align*}
The sequence $\{V_k\}$ converges if and only if 
$\theta' > \theta_c$.
If 
\begin{equation}
\label{epsilon_stability_condition2}
(1-\varepsilon) \theta^*  + \frac{\varepsilon}{M} \sum_{m=1}^M \theta_m > \theta_c,
\end{equation} then one can always find a sufficiently small $\delta > 0$ to satisfy $\theta' > \theta_c$. As \eqref{epsilon_stability_condition2} is equivalent to the condition \eqref{epsilon_stability_condition}, this proves the ``if'' direction of the theorem. 

For the ``only if'' direction, first define
\begin{equation}
\label{eqn:opt_reception_prob_epsilon_greedy}
\tilde{\theta} \triangleq \theta^*(1-\varepsilon) + \frac{\varepsilon}{M} \sum_{m=1}^M \theta_m,
\end{equation}
and note that, for a given $\varepsilon$, 
$$
\mathbb{P}(\textnormal{reception at time } k) \leq \tilde{\theta}$$ 
holds, because $\tilde{\theta}$ is the reception probability assuming the optimal $m^*$ is always chosen during the exploitation phase. 
If $ \varepsilon \geq (\theta^* - \theta_c)/(\theta^* - \frac{1}{M}\sum_{m=1}^M \theta_m)$, then $
\mathbb{P}(\textnormal{reception at time } k) \leq \tilde{\theta} \leq \theta_c $, and 
thus
$$ \mathbb{E}[P_{k+1}] \geq \theta_c \overline{P} + (1-\theta_c) h(\mathbb{E}[P_k]).$$
Define a sequence $\{\tilde{V}_k\}$ by
$$ \tilde{V}_{k+1} =  \theta_c \overline{P} + (1-\theta_c) h(\tilde{V}_k), $$
which from the definition of $\theta_c$ is a divergent sequence.
An induction argument shows that $\mathbb{E}[P_k] \geq \tilde{V}_{k}, \forall k$. This implies that $\mathbb{E}[P_k]$ also diverges and completes the proof. 
\end{proof}

\begin{remark}
If $(\theta^* - \theta_c)/(\theta^* - \frac{1}{M}\sum_{m=1}^M \theta_m) > 1$, then Theorem~\ref{thm:stability_epsilon_greedy} establishes that the $\varepsilon$-greedy algorithm provides estimation stability for all $\varepsilon > 0$. 
\end{remark}

We next consider stability of the sampling based algorithms. 
We first give a preliminary result.
\begin{lemma}
\label{lemma:infinite_exploration}
Under i) Thompson sampling, ii) OBS, and iii) SBS,  all channels will be used infinitely often, and
\begin{equation}
    \label{eqn:persist}
\mathbb{P}(m_k = m^*) \rightarrow 1 \textnormal{ as } k \rightarrow \infty.
\end{equation}
\end{lemma}

\begin{proof}
For Thompson sampling and OBS, this is proved in \cite{MayKorda} (see also \cite{Granmo} for the two armed bandit case under Thompson sampling). 
For SBS we note that, in view of \eqref{eqn:OBS_mod}, it holds that $\hat{\theta}_{m,k}^{\textnormal{TS}} \leq \hat{\theta}_{m,k}^{\textnormal{SBS}} \leq \hat{\theta}_{m,k}^{\textnormal{OBS}}$. Thus, intuitively \eqref{eqn:persist} should also hold, as SBS in a sense lies in between Thompson sampling and OBS (and convergence holds for both of these schemes). A rigorous proof of Lemma~\ref{lemma:infinite_exploration} for SBS can be given using similar arguments as in \cite{MayKorda}. The details are omitted for brevity.
\end{proof}

\begin{theorem}
\label{thm:stability_TS}
Suppose condition \eqref{best_channel_condition} holds. Then under i) Thompson sampling, ii) OBS, and iii) SBS, the expected error covariance $\mathbb{E}[P_k]$ is bounded for all $k$.
\end{theorem}
\begin{proof}
From Lemma \ref{lemma:infinite_exploration}, we know that under all three sampling based schemes, given any $\delta > 0$, there exists a $K(\delta) < \infty$ such that
$$\mathbb{P}(\textnormal{reception at time } k) \geq \theta^*(1-\delta), \forall k > K(\delta).$$
Pick a sufficiently small $\delta$ such that 
$\theta' \triangleq \theta^*(1-\delta) > \theta_c$ still holds. 
Then for $k > K(\delta)$, we have by Lemma \ref{lemma:h_properties} and a similar derivation to \eqref{eqn:EP_recursion} that
$$ \mathbb{E}[P_{k+1}] \leq \theta' \overline{P} + (1-\theta') h(\mathbb{E}[P_k]).$$
Define a sequence $\{V_k\}$ by 
\begin{align*}
V_{k} & = \mathbb{E}[P_k], \quad k \leq K(\delta)\\
V_{k+1} & = \theta'  \overline{P} + (1-\theta') h(V_k), \quad k \geq K(\delta).
\end{align*}
Using an induction argument, we can show that $V_k$ upper bounds $\mathbb{E}[P_k]$ for all $k $. Since $\theta' > \theta_c$, the sequence $\{V_k\}$ converges, and in particular  $V_k$ is bounded for all $k$. Hence $\mathbb{E}[P_k]$ is also bounded for all $k$. 
\end{proof}

\begin{remark}
Suppose the system matrix $A$, and hence $\theta_c$, is not known exactly, but we believe  $\theta_c$ to be equal to $\hat{\theta}_c$. Then in Line~4 of Algorithm \ref{alg:SBS} we would replace $\theta_c$ by $\hat{\theta}_c$. However, even with an incorrect value $\hat{\theta}_c$, we still have $\hat{\theta}_{m,k}^{\textnormal{TS}} \leq \hat{\theta}_{m,k}^{\textnormal{SBS}} \leq \hat{\theta}_{m,k}^{\textnormal{OBS}}$ and Lemma~\ref{lemma:infinite_exploration} would still hold. Hence SBS guarantees stability even for an incorrect $\hat{\theta}_c$, provided condition \eqref{best_channel_condition} holds for the true $\theta_c$. 
\end{remark}

\section{Performance Bounds}
\label{sec:regret_bounds}
In the multi-armed bandit literature, it is common to analyze the performance of an algorithm via the notion of \emph{regret} \cite{Slivkins_book,Lattimore_book}, which is defined as the difference between the expected cumulative reward from always playing the optimal arm and the expected cumulative reward using a particular bandit algorithm.  The $\varepsilon$-greedy algorithm when applied to the standard multi-armed bandit problem is well known to achieve linear regret. More sophisticated algorithms which can achieve logarithmic regret (which, in a sense, is best possible \cite{LaiRobbins}) include UCB \cite{AuerCesaBianchi}, $\varepsilon$-greedy with a carefully chosen decaying exploration probability \cite{AuerCesaBianchi}, and Thompson sampling \cite{AgrawalGoyal_ACM,KaufmannKordaMunos}.

For the remote estimation problem at hand, it is convenient to regard the trace of the estimation error covariance $\textnormal{tr} P_k$ as a \emph{cost}, or alternatively regard $-\textnormal{tr} P_k$ as a \emph{reward}, see also~\cite{LeongRamaswamy}. Motivated by the notion of regret in standard multi-armed bandit problems, in this paper we define the \emph{estimation regret} over a horizon $T$ for the remote state estimation problem as 
\begin{equation}
\label{eqn:regret_defn}
\begin{split}
\textnormal{regret}^E(T) & \triangleq \sum_{k=1}^T \big( -\textnormal{tr} \mathbb{E}[P_k^*] - (-\textnormal{tr}\mathbb{E}[P_k]) \big) \\
& = \sum_{k=1}^T \textnormal{tr} (\mathbb{E}[P_k]-\mathbb{E}[P_k^*]). 
\end{split}
\end{equation}
In \eqref{eqn:regret_defn}, $\mathbb{E}[P_k^*]$ is the expected error covariance assuming that the optimal channel $m^*$ is always chosen, and  $\mathbb{E}[P_k]$ is the expected error covariance when using a particular channel selection algorithm. Thus, \eqref{eqn:regret_defn} constitutes a measure of the degree of suboptimality incurred from using a particular algorithm. Stationary and transient performance of an algorithm can be quantified by inspecting how fast the regret grows in relation to the horizon length $T$.

To state our results, we first recall some order notation.
\begin{definition}
Given functions $f(.)$ and $g(.)$, we say that $f(T) = O(g(T))$ if there exist constants $c$ and $T_0$, such that $|f(T)| \leq c g(T), \forall T \geq T_0$. 
We say that $f(T) = \Omega(g(T))$ if  $g(T) = O(f(T))$, and that $f(T) = \Theta(g(T))$ if both $f(T) = O(g(T))$ and $f(T) = \Omega(g(T))$. We say that $f(T) = o(g(T))$ if $\lim_{T\rightarrow \infty} f(T)/g(T) = 0$.
\end{definition}

For $\varepsilon$-greedy, the estimation regret scales linearly with $T$. 
\begin{theorem}
\label{thm:regret_bound_epsilon_greedy}
Suppose conditions  \eqref{best_channel_condition} and \eqref{epsilon_stability_condition} hold. Then under the $\varepsilon$-greedy algorithm, we have 
$$\textnormal{regret}^E(T) = \Theta (T).$$
\end{theorem}

\begin{proof}
We first show that $\textnormal{regret}^E(T) = O (T)$. 
By Theorem~\ref{thm:stability_epsilon_greedy}, there exists some bounded $\mathcal{P}^\varepsilon$ such that
$$\mathbb{E}[P_k] \leq \mathcal{P}^\varepsilon, \quad \forall k.$$ 
We have
\begin{align*}
\textnormal{regret}^E(T) &  = \sum_{k=1}^T \textnormal{tr} (\mathbb{E}[P_k]-\mathbb{E}[P_k^*]) \\
& \leq \sum_{k=1}^T \textnormal{tr} \mathbb{E}[P_k] \\
& \leq  \textnormal{tr} (\mathcal{P}^\varepsilon) T  = O(T).
\end{align*}

We now show that $\textnormal{regret}^E(T) = \Omega (T)$. Recall the definition $\tilde{\theta}$ from \eqref{eqn:opt_reception_prob_epsilon_greedy}. Note that we have
\begin{align*}
\mathbb{E}[P_k] &\geq \tilde{\theta} \overline{P} + (1-\tilde{\theta}) h(\mathbb{E}[P_{k-1}]) \geq \tilde{\theta} \overline{P} + (1-\tilde{\theta}) h(\mathbb{E}[P_{k-1}^*]), \\
\mathbb{E}[P_k^*] &= \theta^* \overline{P} + (1-\theta^*)  h(\mathbb{E}[P_{k-1}^*]).
\end{align*}
Then 
\begin{align}
 \textnormal{tr} (\mathbb{E}[P_k]-\mathbb{E}[P_k^*]) 
 & \geq \textnormal{tr} \big( (\theta^* - \tilde{\theta}) ( h(\mathbb{E}[P_{k-1}^*]) - \overline{P} )\big) \nonumber \\
 & \geq (\theta^* - \tilde{\theta}) \min_{k \geq 1} \textnormal{tr}\big(  h(\mathbb{E}[P_{k-1}^*]) - \overline{P}\big) \nonumber \\
 &\triangleq r. \label{eqn:r_epsilon_greedy}
\end{align}
By the definition \eqref{eqn:opt_reception_prob_epsilon_greedy} and Lemma \ref{lemma:h_properties}, we have $r > 0$. Hence
\begin{align*}
\textnormal{regret}^E(T) &  = \sum_{k=1}^T \textnormal{tr} (\mathbb{E}[P_k]-\mathbb{E}[P_k^*]) \\
 & \geq r T  = \Omega (T).
\end{align*}

\end{proof}

For the Bernoulli bandit problem with per stage rewards in $\{0,1\}$, it is known that the regret scales logarithmically with $T$  under Thompson sampling \cite{AgrawalGoyal_ACM,KaufmannKordaMunos}. An ``age-of-information regret'' measure was recently considered in \cite{Fatale_AOI}, where the per stage cost is unbounded, and can increase by (at most) one at every time step. It is shown that this age-of-information regret also scales logarithmically with $T$. In the current work, it follows from \eqref{eqn:h_defn} that the per stage cost $\mathbb{E}[P_k]$ (or reward $-\mathbb{E}[P_k]$) may also become unbounded, and furthermore can increase at an exponential rate. Interestingly, for the estimation regret introduced in  \eqref{eqn:regret_defn}, a bound that is logarithmic in $T$ still holds. 

Let us refer to sub-optimal channels as those whose packet reception probabilities are not equal  to the optimal $\theta^*$.
We begin with a preliminary result.
\begin{lemma}
\label{lemma:suboptimal_arms}
Let $N_{\textnormal{sub}}(T) $ denote the number of uses of sub-optimal channels over horizon $T$. Then, under i) Thompson sampling, ii) OBS, and iii) SBS, we have
$$\mathbb{E}[N_{\textnormal{sub}}(T)] = \Theta(\log T).$$
\end{lemma}

\begin{proof}
The property that $\mathbb{E}[N_{\textnormal{sub}}(T)] = O(\log T)$ is shown for Thompson sampling in \cite{AgrawalGoyal_ACM,KaufmannKordaMunos}. The corresponding result for OBS is proved in \cite{Mellor_thesis}, based on similar arguments as in \cite{AgrawalGoyal_ACM}. For SBS, $\mathbb{E}[N_{\textnormal{sub}}(T)] = O(\log T)$  can be shown by making use of the relation $\hat{\theta}_{m,k}^{\textnormal{TS}} \leq \hat{\theta}_{m,k}^{\textnormal{SBS}} \leq \hat{\theta}_{m,k}^{\textnormal{OBS}}$ and the arguments of \cite{Mellor_thesis}. The (rather lengthy) details are omitted for brevity. 

We next recall a result from \cite{LaiRobbins}, namely that any policy with $\mathcal{R}(T) = o(T^a)$ for every $a>0$ satisfies $\mathbb{E}[N_{\textnormal{sub}}(T)] = \Omega(\log T)$, where  $\mathcal{R}(T) \triangleq \mathbb{E} \big[\sum_{k=1}^T (\theta^* - \theta(k)) \big]$ is the classical notion of regret for the Bernoulli bandit problem, see e.g. \cite{AgrawalGoyal_ACM,KaufmannKordaMunos}. Since for Thompson sampling, OBS, and SBS, we have
$$ \mathcal{R}(T) = O\big(\mathbb{E}[N_{\textnormal{sub}}(T)]\big) = O(\log T) = o(T^a), \, \forall a > 0,$$
where the first equality follows from standard arguments \cite{AgrawalGoyal_ACM,KaufmannKordaMunos} and the second equality is what we have just shown, the property $\mathbb{E}[N_{\textnormal{sub}}(T)] = \Omega(\log T)$ therefore holds for all three schemes. 
\end{proof}

\begin{figure}[t]
	\centering\includegraphics[scale=0.87]{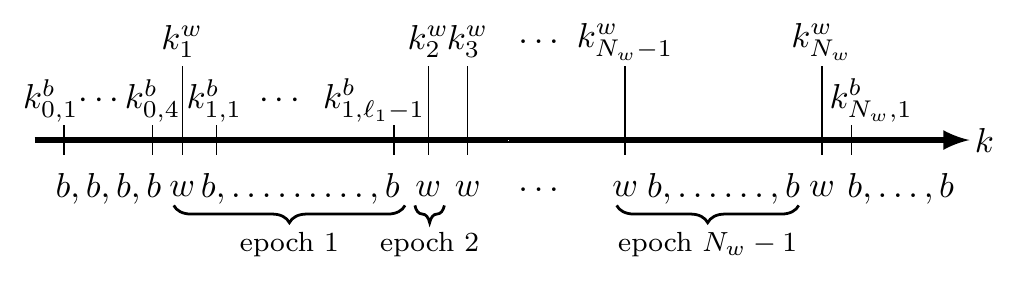}
	\caption{Channel selection based on schedule $\mathcal{A}$, where `w' and `b' refer to the worst and the best channel selection, respectively.}
	\label{fig:epochs}
\end{figure}

\begin{theorem}
\label{thm:regret_bound_general}
Suppose that condition \eqref{best_channel_condition} holds. Then under i) Thompson sampling, ii) OBS, and iii) SBS, we have
$$ \textnormal{regret}^E(T)  = \Theta (\log T). $$
\end{theorem}

Before giving the formal proof of this result, we first provide a sketch of the proof strategy for the upper bound $ \textnormal{regret}^E(T)  = O(\log T)$. Following the idea of \cite{Fatale_AOI}, which showed a logarithmic bound for the age-of-information regret, we will upper bound the estimation regret with the regret that is achieved using an alternative schedule, say $\mathcal{A}$, that replaces all uses of sub-optimal channels with the worst channel. In \cite{Fatale_AOI}, this schedule $\mathcal{A}$ is then further upper bounded by a ``worst-case'' schedule $\mathcal{B}$ that groups all the uses of the worst channel together. However, for the problem at hand, as the expected error covariance can increase exponentially fast if the worst channel does not stabilize the remote estimator, this argument will give the desired logarithmic regret bound only in problem instances where the worst channel (and hence all channels) is stabilizing. To cover more interesting and challenging situations, to establish Theorem~\ref{thm:regret_bound_general}, we will consider a division of the time interval into ``epochs'' separated by successive uses of the worst channel, see Fig.~\ref{fig:epochs}. Through a careful analysis, we will show that the accumulated regret within each epoch is bounded, while the expected number of epochs is logarithmic in $T$. As a consequence the required logarithmic regret upper bound is established. 
The detailed proof of Theorem \ref{thm:regret_bound_general} now follows.

\begin{proof}
We will first prove that $\textnormal{regret}^E(T)  = O (\log T)$.
Define the reception probability of the worst channel as $$\theta_w \triangleq \min_m \theta_m.$$ 
Given any schedule of channel selections, let schedule $\mathcal{A}$ denote an alternative schedule which replaces all uses of sub-optimal channels with the worst channel. Let $\mathbb{E}[P_{k}^\mathcal{A}] $ be the expected error covariance under this replacement procedure.  Then from \eqref{eqn:EP_recursion} and Lemma \ref{lemma:h_properties}, we can easily show that $\mathbb{E}[P_k] \leq \mathbb{E}[P_k^\mathcal{A}], \, \forall k.$
In particular,
\begin{equation}
\label{eqn:EPk_sum_bound}
\sum_{k=1}^T \textnormal{tr} \mathbb{E}[P_k] \leq \sum_{k=1}^T \textnormal{tr} \mathbb{E}[P_k^\mathcal{A}].    
\end{equation}
Then, we have
\begin{equation}\label{eq:regretA}
\text{regret}^E(T)\leq \text{regret}^{\mathcal{A}}(T) \triangleq \sum_{k=1}^{T} \tr\left(\myexpect{P^{\mathcal{A}}_k} - \myexpect{P^*_k} \right).
\end{equation}
Thus, we only need to prove that $\text{regret}^{\mathcal{A}}(T)=O(\log T)$.

Let $k^w_i \in \mathbb{N}, i \geq 1$ denote the time index of the $i$th usage of the worst channel based on schedule $\mathcal{A}$.
The  time horizon is divided by $\{k^w_i\}$ into epochs, as illustrated in Fig.~\ref{fig:epochs}.
The $i$th epoch starts from $k^w_i$ and ends before $k^w_{i+1}$, and is of epoch length $\ell_i \triangleq k^w_{i+1} - k^w_{i} \geq 1$. Note that $\ell_i$ is a random process with a time-varying distribution due to the bandit algorithms.
We define 
\begin{equation}\label{eq:Delta_w}
\Delta^w_i \triangleq \myexpect{P^{\mathcal{A}}_{k^w_i}} - \myexpect{P^*_{k^w_i}} 
\end{equation}
as the gap between the expected estimation error covariances achieved by schedule $\mathcal{A}$ and the persistent schedule of the best channel.

Let  $k^b_{i,j} \in \mathbb{N}, i,j \geq 1$ denote the time  index based on schedule $\mathcal{A}$ of the $j$th usage of the best channel  in epoch $i$, when the epoch length $\ell_i>1$ (see Fig.~\ref{fig:epochs}).
Let $k^b_{0,j} \in \mathbb{N}, j \geq 1$ denote the time  index of the $j$th usage of the best channel before the first epoch and $\ell_0$ denote the total number of time slots before the first epoch.
Then, we define 
\begin{equation}\label{eq:Delta_b}
\Delta^b_{i,j} \triangleq \myexpect{P^{\mathcal{A}}_{k^b_{i,j}}} - \myexpect{P^*_{k^b_{i,j}}}, i\geq 0, j \geq 1, \text{ if } \ell_i>1.
\end{equation}
In particular, we have $\Delta^b_{i,j} =0$ when $i=0$, since schedule $\mathcal{A}$ is identical to the persistent schedule of the best channel before the first epoch.

We use random variable $N_w$ to denote the number of uses of the worst channel within a time interval of length $T$. Then, $\text{regret}^{\mathcal{A}}(T)$ in \eqref{eq:regretA} can be upper bounded as 
\begin{equation}\label{eq:regretA2}
\begin{aligned}
\text{regret}^{\mathcal{A}}(T) \!&\leq \!
\mathbb{E}\!\! \left[\sum_{j=1}^{\ell_0-1}\!\!\tr\Delta^b_{0,j} + \!\sum_{i=1}^{N_w} \left(\tr\Delta^w_i+\sum_{j=1}^{\ell_i-1}\!\tr\Delta^b_{i,j}\right) \right]\!\!\\
&= \mathbb{E}\!\! \left[\!\sum_{i=1}^{N_w} \left(\tr\Delta^w_i+\sum_{j=1}^{\ell_i-1}\!\tr\Delta^b_{i,j}\right) \right],
\end{aligned}
\end{equation}
where the expectation is over $N_w$ and $\{\ell_i\}$.
To prove the desired result, we will analyze $\left(\tr\Delta^w_i+\sum_{j=1}^{\ell_i-1}\tr\Delta^b_{i,j}\right)$ first.

\emph{Analysis of per epoch regret}: By applying the estimation error covariance updating rule \eqref{remote_estimator_eqns} into \eqref{eq:Delta_b}, we have
\begin{equation}\label{eq:Delta_b_analysis1}
\begin{aligned}
\Delta^b_{i,j} &=  (1-\theta^*) h(\myexpect{P^{\mathcal{A}}_{k^b_{i,j}-1}}) + \theta^* \overline{P} \\
& \quad - \left((1-\theta^*) h(\myexpect{P^{*}_{k^b_{i,j}-1}}) + \theta^* \overline{P}\right)\\
&=(1-\theta^*) A \left(\myexpect{P^{\mathcal{A}}_{k^b_{i,j}-1}} - \myexpect{P^{*}_{k^b_{i,j}-1}}\right)  A^\top.
\end{aligned}
\end{equation}
For the time slots in epoch $i$ when $\ell_i>1$, it directly follows that
\begin{equation}\label{eq:Delta_b_analysis2}
\Delta^b_{i,j} = (1-\theta^*)^{j} A^{j} \Delta^w_{i} (A^{j})^\top, 1\leq j \leq \ell_i-1.
\end{equation}
Applying the trace operator to inequality \eqref{eq:Delta_b_analysis2}, it can be obtained that
\begin{align}\label{eq:trace_w_1}
\text{tr} \Delta^b_{i,j}
&= (1-\theta^*)^{j} \text{tr}  \left(A^{j} \Delta^w_{i} (A^{j})^\top\right)\\
\label{eq:trace_w_2}
&\leq  (1-\theta^*)^{j} \text{tr}  ((A^{j})^\top A^{j}) \text{tr}\Delta^w_{i}\\
\label{eq:trace_w_3}
&\leq   \kappa_1(\epsilon)  (1-\theta^*)^{j}  (\rho(A)+\epsilon)^{2j}
\text{tr}\Delta^w_{i}
\end{align}
where $\epsilon > 0$ is arbitrary and $\kappa_1 (\epsilon) > 0$. Inequality \eqref{eq:trace_w_2} is due to the property that $\text{tr}(UV) \leq \text{tr}(U) \text{tr}(V)$ for any positive semidefinite matrices $U$ and $V$, see e.g. \cite[p.445]{HornJohnson}.
Inequality \eqref{eq:trace_w_3} is due to the fact that $\text{tr}  ((A^{j})^\top A^{j})$ is the sum of squares of all elements of $A^{j}$, and is obtained based on the lemma below about the element-wise bounds of matrix powers~\cite[Lemma 2]{liu2020remote}:
\begin{lemma}\label{lem:Liu}
	Consider a $z$-by-$z$ matrix $Z$, and let $[Z]_{i,i'}$ denote the element at the $i$th row and $i'$th column of $Z$. \par
	(i) For any $\epsilon>0$, there exists a $\kappa_1(\epsilon)>0$ such that $\big|[Z^j]_{i,i'}\big|^2 \leq \kappa_1(\epsilon) (\rho(Z)+\epsilon)^{2j}, \forall i,i'\in \{1,\dots,z\}$.\par
	(ii) There exist $\kappa_2>0$ and $L\in\mathbb{N}$ such that $\big|[Z^j]_{i,i'}\big|^2$ is periodically lower bounded by $\kappa_2 \rho(Z)^{2j}$ with period $L$. 
\end{lemma}

From assumption \eqref{best_channel_condition}, we can find a sufficiently small $\epsilon$ such that $(1-\theta^*)(\rho(A)+\epsilon)^2<1$.
Using the inequality~\eqref{eq:trace_w_3}, we obtain  
\begin{equation}\label{eq:sum_trace_b}
\sum_{j=1}^{\ell_i-1}\tr\Delta^b_{i,j} \leq  \frac{\kappa_1(\epsilon) \text{tr}\Delta^w_{i}}{1- (1-\theta^*) (\rho(A)+\epsilon)^2}.
\end{equation}
Therefore, the regret of epoch $i$ is bounded as
\begin{equation}\label{eq:Deltawb_bound}
\tr\Delta^w_i+\sum_{j=1}^{\ell_i-1}\!\tr\Delta^b_{i,j} \leq \kappa(\epsilon) \tr\Delta^w_i,
\end{equation}
where $\kappa(\epsilon) \triangleq 1+ \frac{\kappa_1(\epsilon)}{1- (1-\theta^*) (\rho(A)+\epsilon)^2}  >0 $.

\emph{Analysis of $\textnormal{tr}\Delta^w_{i}$}:
By applying the estimation error covariance updating rule \eqref{remote_estimator_eqns} into \eqref{eq:Delta_w}, we have
\begin{equation} \label{eq:Delta_w_analysis}
\begin{aligned}
\Delta^w_i 
&= (1-\theta_w)h(\myexpect{P^{\mathcal{A}}_{k^w_i-1}}) + \theta_w \overline{P} \\
& \quad - \left((1-\theta^*)h(\myexpect{P^{*}_{k^w_i-1}}) + \theta^* \overline{P} \right) \\
&\leq  (1-\theta_w)h(\myexpect{P^{\mathcal{A}}_{k^w_i-1}}) + \theta_w \overline{P} \\
&= (1-\theta_w)h(\myexpect{P^{\mathcal{A}}_{k^w_i-1}}-\myexpect{P^{*}_{k^w_i-1}}) \\ & \quad + (1-\theta_w) h(\myexpect{P^{*}_{k^w_i-1}}) + \theta_w \overline{P} \\
&\leq  (1-\theta_w)A(\myexpect{P^{\mathcal{A}}_{k^w_i-1}}-\myexpect{P^{*}_{k^w_i-1}})A^\top + G
\end{aligned}
\end{equation}
where $G\triangleq(1-\theta_w) \left(A \mathcal{P} A^\top + Q \right) + \theta_w \overline{P} $, and $\mathcal{P}$ is a bound on $\mathbb{E}[P^*_k]$ which exists due to the stability condition \eqref{best_channel_condition}.

Next, using a similar argument as in the proof of Theorem~\ref{thm:stability_TS}, we note that the expected error covariances $\{\myexpect{P^\mathcal{A}_k}\}$ are bounded by a positive semidefinite matrix $\mathcal{P}'$ for all $k$. From \eqref{eq:Delta_w_analysis}, it directly follows that
$\Delta^w_i \leq (1-\theta_w)A \mathcal{P}' A^\top + G$. Then, we can find a positive constant $\xi$ such that
\begin{equation}\label{eq:Deltaw_bound}
\tr \Delta^w_i \leq \xi, \forall i.
\end{equation}

By taking \eqref{eq:Deltawb_bound} and \eqref{eq:Deltaw_bound} into \eqref{eq:regretA2} and using Lemma \ref{lemma:suboptimal_arms}, 
$$
\begin{aligned}
\text{regret}^E(T) &\leq \text{regret}^{\mathcal{A}}(T)  \\
&\leq 
\mathbb{E}\!\!\left[\sum_{i=1}^{N_w} \kappa(\epsilon) \xi\right]= \kappa(\epsilon) \xi \myexpect{N_w} = O(\log T).
\end{aligned}
$$

To conclude, we now prove that $\textnormal{regret}^E(T)  = \Omega (\log T)$. We have
\begin{align*}
\textnormal{regret}^E(T) & =  \sum_{k=1}^{T} \textnormal{tr} (\mathbb{E}[P_k] -  \mathbb{E}[P_k^*]) \\
& = \mathbb{E} \bigg[\sum_{k=1}^{T} \textnormal{tr} (\mathbb{E}[P_k] -  \mathbb{E}[P_k^*]) \mathds{1}(\theta(k) = \theta^*) \\
& \quad + \sum_{k=1}^{T} \textnormal{tr} (\mathbb{E}[P_k] -  \mathbb{E}[P_k^*]) \mathds{1}(\theta(k) \neq \theta^*) \bigg] \\
& \geq \mathbb{E} \bigg[ \sum_{k=1}^{T} \textnormal{tr} (\mathbb{E}[P_k] -  \mathbb{E}[P_k^*]) \mathds{1}(\theta(k) \neq \theta^*) \bigg].
\end{align*}
Let $\theta^\circ$ denote the second largest packet reception probability. When $\theta(k) \neq \theta^*$, we can show similar  to \eqref{eqn:r_epsilon_greedy}  that 
\begin{align*}
\textnormal{tr} (\mathbb{E}[P_k]-\mathbb{E}[P_k^*]) 
& \geq \textnormal{tr} \big( (\theta^* - \theta^\circ) ( h(\mathbb{E}[P_{k-1}^*]) - \overline{P} )\big) \nonumber \\
& \geq (\theta^* - \theta^\circ) \min_{k \geq 1} \textnormal{tr}\big( h(\mathbb{E}[P_{k-1}^*]) - \overline{P}\big) \nonumber \\
&\triangleq r^\circ > 0.
\end{align*}
Thus 
\begin{align*}
\textnormal{regret}^E(T) & \geq r^\circ \mathbb{E} \bigg[ \sum_{k=1}^T \mathds{1}(\theta(k) \neq \theta^*) \bigg] \\ 
& = r^\circ \mathbb{E}[N_{\textnormal{sub}}(T)] \\
& = \Omega (\log T),
\end{align*}
where the last line follows from Lemma \ref{lemma:suboptimal_arms}. 
This completes the proof that $\textnormal{regret}^E(T)  = \Omega (\log T)$, and hence establishes Theorem \ref{thm:regret_bound_general}.
\end{proof}

\begin{remark}
We note that some works on learning based control have made  assumptions that the underlying system is stable \cite{Lale}, or can be stabilized for all possible controls \cite{OuyangGagraniJain}, when carrying out a regret analysis. This stands in contrast to  Theorem~\ref{thm:regret_bound_general}, stated above, which merely  requires that at least one of the $M$ channels is stabilizing for the remote estimator. Our result thus considers scenarios where some of the sub-optimal channels can cause the expected error covariance to diverge if they are used too often. 
\end{remark}

\section{Numerical Studies}
\label{sec:numerical}

\begin{figure}[t!]
\centering 
\includegraphics[scale=0.3]{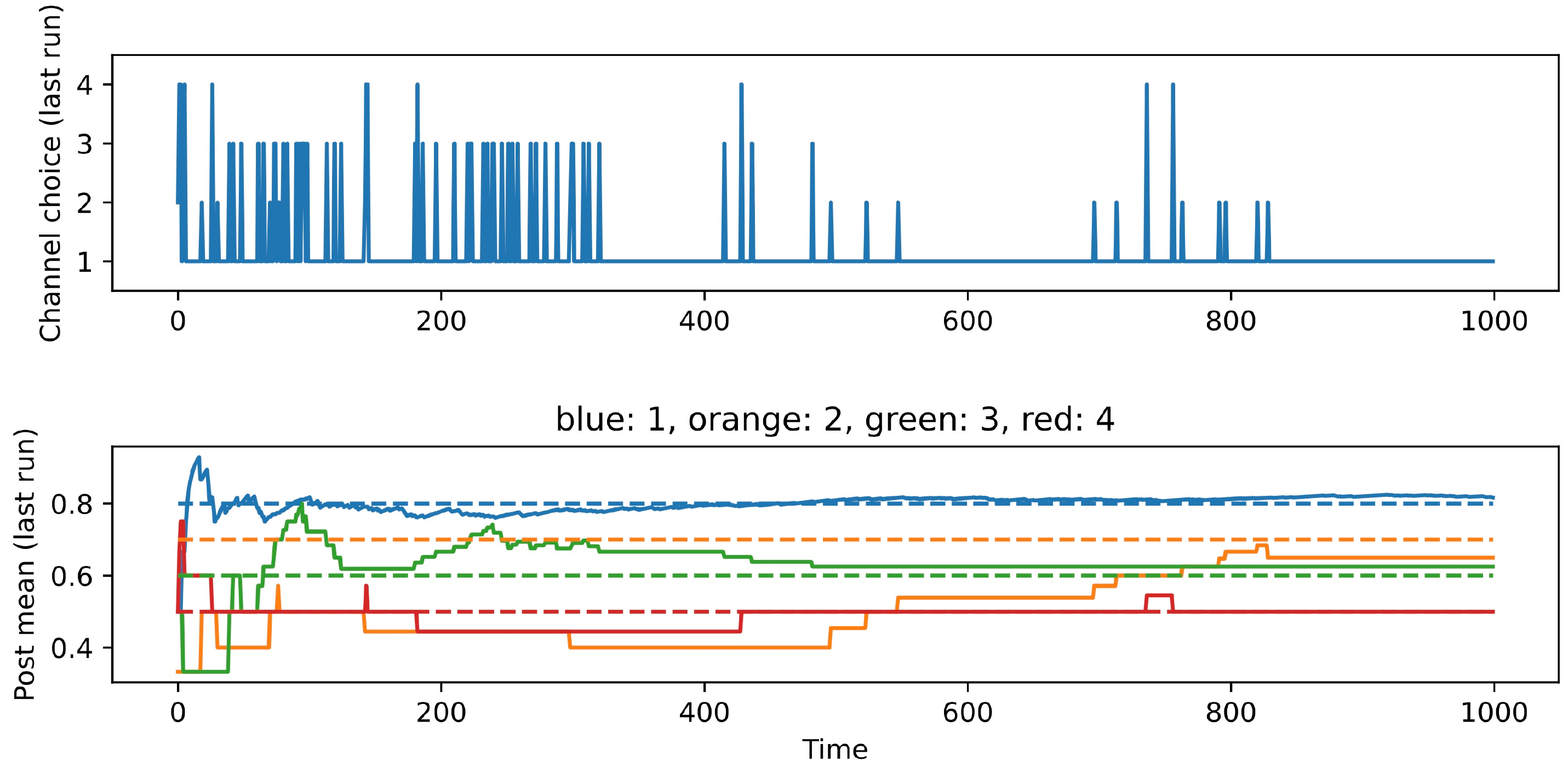} 
\caption{Learning using Stability-aware Bayesian Sampling}
\label{fig:learning_plot}
\end{figure} 

We consider a situation with $M=4$ channels. Figure~\ref{fig:learning_plot} illustrates how Stability-aware Bayesian sampling (SBS) learns, and increasingly uses, the optimal channel, for a situation where 
$$  A = \begin{bmatrix} 1.5 & 0.2 \\ 0.3 & 0.9\end{bmatrix}, \quad C = \begin{bmatrix} 1 & 1\end{bmatrix}, \quad Q=I, \quad R=1$$
and the (unknown) channel success probabilities are given by
$$\bm{\theta}=(\theta_1, \theta_2, \theta_3, \theta_4) =(0.8, \; 0.7, \; 0.6, \; 0.5).$$
Table \ref{table:algorithm_comparison} illustrates the empirical estimation regret for a number of different channel probabilities $\bm{\theta}=(\theta_1, \theta_2, \theta_3, \theta_4)$ and system matrices $A$, $C$, $Q$, $R$. All bandit algorithms studied in this work, namely $\varepsilon$-greedy, Thompson sampling (TS), OBS, and SBS, are considered. The regret is computed over horizon $T = 1000$, and averaged over 100000 runs. For $\varepsilon$-greedy, we considered different values of $\varepsilon$ in steps of 0.02, with the value giving the smallest simulated regret presented in Table~\ref{table:algorithm_comparison}.

\begin{table*}[t!]
\caption{Estimation regret of different algorithms}
\centering
\begin{tabular}{|c|c|c|c|c|c|c||c|c|c|c|} \hline
$\bm{\theta}$ & $A$ & $C$ & $Q$ & $R$ & $\theta_c$ & $\varepsilon$ & $\varepsilon$-greedy & TS & OBS & SBS \\ \hline \hline 
$(0.8, 0.75, 0.55, 0.5)$ & 1.45  & 1 & 1 & 1 & 0.524 & 0.12 & 263 & 143 & 136 & \textbf{135} \\ \hline 
$(0.7, 0.6, 0.5, 0.4)$ & 1.45  & 1 & 1 & 1 & 0.524 & 0.18 & 996 & 679 & \textbf{596} & 603 \\
\hline
$(0.6, 0.5, 0.4, 0.3)$ & 1.45  & 1 & 1 & 1 & 0.524 & 0.22 & 10319 & 8253 & 6823 & \textbf{6276} \\
\hline
$(0.7, 0.6, 0.4, 0.3)$ & $\begin{bmatrix} 1.2 & 0.1 \\ 0.2 & 1.1 \end{bmatrix} $ & $\begin{bmatrix} 1 & 1\end{bmatrix}$ & $I$ & 1 & 0.408 & 0.10 & 582 & 295 & 274 & \textbf{273} \\ \hline
$(0.65, 0.55, 0.45, 0.35)$ & $\begin{bmatrix} 1.2 & 0.1 \\ 0.2 & 1.1 \end{bmatrix} $ & $\begin{bmatrix} 1 & 1\end{bmatrix}$ & $I$ & 1 & 0.408 & 0.12 & 886 & 517 & \textbf{473} & 483 \\ \hline
$(0.55, 0.45, 0.35, 0.25)$ & $\begin{bmatrix} 1.2 & 0.1 \\ 0.2 & 1.1 \end{bmatrix} $ & $\begin{bmatrix} 1 & 1\end{bmatrix}$ & $I$ & 1 & 0.408 & 0.18 & 4723 & 3424 & 2727 & \textbf{2430} \\ \hline
$(0.9, 0.8, 0.7, 0.5)$ & $\begin{bmatrix} 1.5 & 0.2 \\ 0.3 & 0.9 \end{bmatrix} $ & $\begin{bmatrix} 1 & 1\end{bmatrix}$ & $I$ & 1 & 0.603 & 0.14 & 290 & 104 & 98 & \textbf{97} \\ \hline
$(0.8, 0.7, 0.6, 0.5)$ & $\begin{bmatrix} 1.5 & 0.2 \\ 0.3 & 0.9 \end{bmatrix} $ & $\begin{bmatrix} 1 & 1\end{bmatrix}$ & $I$ & 1 & 0.603 & 0.18 & 803 & 403 & \textbf{354} & 363 \\ \hline
$(0.7, 0.6, 0.5, 0.4)$ & $\begin{bmatrix} 1.5 & 0.2 \\ 0.3 & 0.9 \end{bmatrix} $ & $\begin{bmatrix} 1 & 1\end{bmatrix}$ & $I$ & 1 & 0.603 & 0.22 & 7185 & 6178 & 3903 & \textbf{3106} \\ \hline
\end{tabular}
\label{table:algorithm_comparison}
\end{table*}

As expected, the $\varepsilon$-greedy algorithm is outperformed by the sampling based algorithms.   
In all the considered scenarios OBS and SBS can achieve at least some improvement in performance over Thompson sampling.\footnote{Note that when looking at individual simulation runs, we cannot say that any one scheme will outperform any other.}  In many scenarios the performance of OBS and SBS are quite close to each other. However in more challenging situations, where there are only few stabilizing channels (just one or two), SBS seems to be able to outperform OBS. This could be due to the fact that  in these cases using non-stabilizing channels more than necessary can significantly degrade performance.

In Fig. \ref{fig:regret_comparison_plot} we plot the estimation regret over time for the $\varepsilon$-greedy and Thompson sampling algorithms. Plots for OBS and SBS are qualitatively similar to Thompson sampling and are omitted. The parameters used are the same as those mentioned at the start of Section \ref{sec:numerical}.  We see a linear scaling (at larger times) for $\epsilon$-greedy and a logarithmic scaling in the case of Thompson sampling, in agreement with Theorems \ref{thm:regret_bound_epsilon_greedy} and \ref{thm:regret_bound_general}. 
\begin{figure}[t!]
\centering 
\includegraphics[scale=0.5]{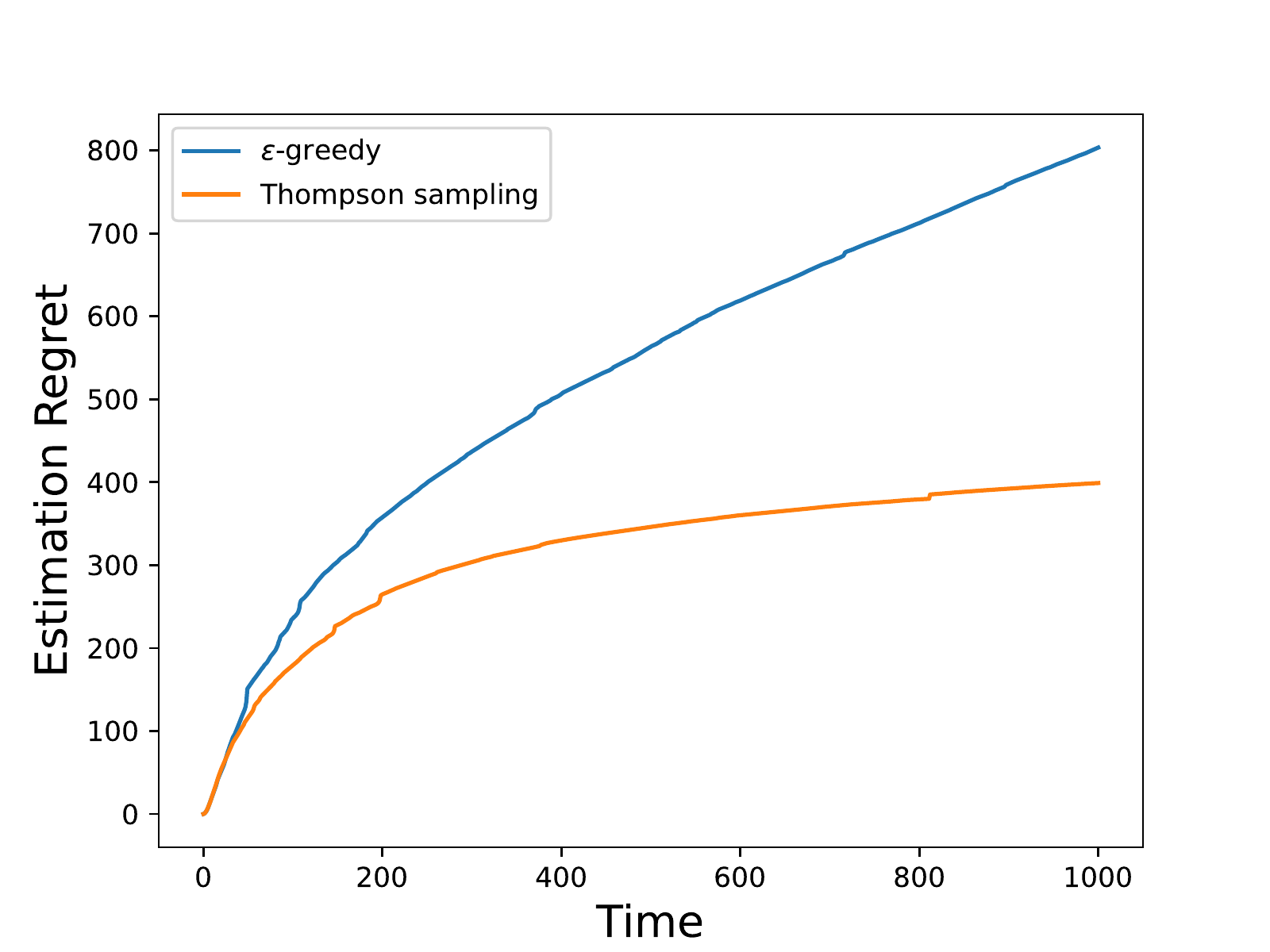} 
\caption{Estimation regret over time}
\label{fig:regret_comparison_plot}
\end{figure}

\section{Conclusion}
\label{sec:conclusion}
We have studied a remote state estimation problem where a sensor can choose between a number of channels with unknown channel statistics, with which to transmit over. We have made use of bandit algorithms to learn the statistics of the channels, while simultaneously carrying out the estimation procedure. Stability of the estimator using these algorithms has been shown, and bounds on the estimation regret have been derived.  
Future work will include the use of bandit algorithms in the allocation of channels to multiple processes.  



\bibliography{IEEEabrv,sensor_scheduling,control}
\bibliographystyle{IEEEtran}

\end{document}